\documentclass[12pt]{article}

\usepackage{graphicx}
\usepackage{wrapfig}
\usepackage{amssymb,amsfonts,amsmath,amsthm,amscd}
\usepackage{color}
\usepackage[english]{babel}
\usepackage[latin1]{inputenc}
\usepackage{times}
\usepackage{geometry}
\usepackage[breaklinks,pdfstartview=FitH,colorlinks,linktocpage,bookmarks=false]{hyperref}
\usepackage{newtxtext,newtxmath}
\usepackage{pgf,tikz,circuitikz}\usetikzlibrary{arrows}
\usepackage{mathrsfs}

\newtheorem{theorem}{Theorem}%[section]
\newtheorem{proposition}[theorem]{Proposition}

\newtheorem{lemma}[theorem]{Lemma}
\newtheorem{conjecture}[theorem]{Conjecture}
\newtheorem{corollary}[theorem]{Corollary}

\theoremstyle{remark}

\newtheorem{remark}[theorem]{Remark}

\theoremstyle{definition}

\newtheorem{definition}[theorem]{Definition}
\newtheorem{example}[theorem]{Example}

\long\def\comment#1\endcomment{}
\long\def\tmpcomment#1\endtmpcomment{#1}

\begin{document}

%\comment

\title{Lattice gauge theory and a random-medium Ising model}
%\title{Simultaneous percolation}

\author{M. Skopenkov}
\date{}

\maketitle

\begin{abstract}
We study linearization of lattice gauge theory. Linearized theory approximates lattice gauge theory in the same manner as the loop O(n)-model approximates the spin O(n)-model. Under mild assumptions, we show that the expectation of an observable in linearized Abelian gauge theory coincides with the expectation in the Ising model with random edge-weights. We find a similar relation between Yang-Mills theory and 4-state Potts model. For the latter, we introduce a new observable.
%All the models are going to be defined; no background in physics is assumed.

\comment
Percolation, lattice gauge theory, and random graphs
M. Skopenkov
NRU Higher School of Economics
Institute for Information Transmission Problems RAS

A model of percolation is a uniform random coloring of the hexagons of a honeycomb lattice into 2 colors. We show that such percolation appears naturally in a linearization of lattice gauge theory (used for computations in particle physics). Linearized gauge theory approximates lattice gauge theory in the same manner as the loop O(n)-model approximates the spin O(n)-model.
We prove a combinatorial identity between expectations of the same observable in the linearized Abelian gauge theory and the Ising model with random edge-weights. We find a similar relation between non-Abelian Yang-Mills theory and
uniform random coloring into 4 colors (the Potts model). We also introduce a new observable in the latter model. All the models are going to be defined; no background in physics is assumed.
\endcomment

\smallskip


%\vspace{-1cm}
\end{abstract}

\footnotetext[0]{The work is supported by Ministry of Science and Higher Education of the Russian Federation, agreement N075-15-2019-1619.
%The publication was prepared within the framework of the Academic Fund Program at the National Research University Higher School of Economics (HSE) in 2018-2019 (grant N18-01-0023) and by the Russian Academic Excellence Project ``5-100''. The author has also received support from the Simons--IUM fellowship.
%The research was carried out at the IITP RAS at the expense of the Russian Foundation for Sciences (project N14-50-00150).
%the President of the Russian Federation grant MK-6137.2016.1, ``Dynasty'' foundation, and the Simons--IUM fellowship.
}

%{%\small
%\tableofcontents
%}

%\section{Introduction}\label{s:intro}

\begin{tabular}{p{0.40\textwidth}p{0.50\textwidth}}
& \textit{Dedicated to the last real scientists, brave to face real difficulties, not sweeping them under the rug.}
\end{tabular}

%%%Dedicated to the last real scientists, who can afford to avow their mistakes. %being wrong.

\bigskip

%Proposition~\ref{prop-linearization}

We prove a combinatorial identity meaning that the expectation of an observable in (pre)linearized $2$-dimensional Abelian lattice gauge theory equals its expectation in the Ising model with random edge-weights, and a similar identity for Yang-Mills theory (see Proposition~\ref{prop-linearization-intro}; all the notions are defined below).

This allows to translate known deep results on the Ising model and its relatives \cite{Smirnov-01,KS-20} to linearized lattice gauge theory. E.g. we provide an observable in the latter, having conformally invariant continuum limit (see Corollary~\ref{cor-assignment} and Remark~\ref{rem-conformal}).
\emph{Linearized} lattice gauge theory (see Definitions~\ref{def-lattice-gauge}, \ref{def-linearized-gauge}, \ref{def-linearized-gauge-general}) approximates the ordinary one in the same manner as the loop O(n)-model approximates the spin O(n)-model \cite[\S1]{Duminil-Copin-etal-14}:  first the exponential weight is replaced by a linear one, and then the measure is simplified without changing the partition function (as well as the expectation of observables having certain reflection symmetries).

Thus our identity extends the ones known in the O(n)-model \cite[Appendix~A]{Duminil-Copin-etal-14}. In Abelian gauge theory, a related Banks--Myerson--Kogut identity was used by A.~Guth to prove a phase transition \cite{Guth-80}. For the Yang-Mills case, we could not find any analoguous identities in the literature.

%Let us state the main results precisely.

%In particular, to each percolation observable we assign a linearized-gauge-theory observable having the same expectation and vice versa. We start with examples and finish with more abstract general propositions.

\begin{proposition}\label{prop-linearization-intro} (See Figure~\ref{fig2}) Let $\Omega$ be a polygon triangulated by regular triangles, $\beta\in\mathbb{R}$. Let $F,E,V$ be the sets of faces, nonboundary edges, nonboundary vertices respectively; $F,E,V\ne\emptyset$. Let $f$ and $g$ be real-valued integrable functions on $U(1)^E$ and $SU(2)^E$ respectively, invariant under the reflection in each coordinate hyperplane of $\mathbb{R}^{2{|E|}}\supset U(1)^E$ and $\mathbb{R}^{4{|E|}}\supset SU(2)^E$. Then
\vspace{-0.2cm}
  \begin{multline*}
  %\begin{align*}
  %%%
  %\hspace{-1cm}\int\limits_{[0,2\pi]^{E}} d\theta \xi(e^{i\theta})\prod_{ABC}  \cos(\theta(AB)+\theta(BC)+\theta(CA))
  %&=
  %\int\limits_{[0,2\pi]^{E}}  d\theta \xi(e^{i\theta})\sum_{\sigma\in \{+1,-1\}^V}\prod_{AB:\Sigma(A)=\Sigma(B)}   \cos^2\theta(AB)
  %\prod_{AB:\sigma(A)\ne\sigma(B)} \sin^2\theta(AB)\\
  %%\,\prod_{AB}dU(AB).
  %\hspace{-1cm}\int\limits_{SU(2)^{E}} dU\eta(U)\prod_{ABC}\mathrm{Re}\,(U(AB)U(BC)U(CA))
  %&=
  %\int\limits_{SU(2)^{E}}  dU\eta(U)\sum_{H\in \{1,i,j,k\}^V}\prod_{AB}\mathrm{Re}^2\,(H(A)^*U(AB)H(B)).
  %%\,\prod_{AB}dU(AB).
  %%%%%%%
  \int\limits_{U(1)^{E}} dU\,f(U)
  \prod_{ABC\in F}
  \left(1+\beta
  \,\mathrm{Re}\,(U(AB)U(BC)U(CA))
  \right)
  %\int\limits_{[0,2\pi]^{E}} d\theta \xi(e^{i\theta})\prod_{ABC}\left(1+\beta
  %\cos(\theta(AB)+\theta(BC)+\theta(CA))
  %\right)%\,\prod_{AB}dU(AB)
  =\\[-0.7cm]=
  \int\limits_{[0,2\pi]^{E}}  d\theta\, f(e^{i\theta})\left(1+\beta^{|F|}\,\sum_{\sigma\in \{+1,-1\}^V}
  \prod_{\substack{AB\in E:\\\sigma(A)=\sigma(B)}}   \cos^2\theta(AB)
  \prod_{\substack{AB\in E:\\\sigma(A)\ne\sigma(B)}}  \sin^2\theta(AB)\right);
  %%\,\prod_{AB}dU(AB).
  \end{multline*}
  and
  \vspace{-0.4cm}
  \begin{multline*}
  \int\limits_{SU(2)^{E}} dU\,g(U)\prod_{ABC\in F} \left(1+\beta
  \,\mathrm{Re}\,(U(AB)U(BC)U(CA))
  \right)%\,\prod_{AB}dU(AB)
  =\\[-0.7cm]=
  \int\limits_{SU(2)^{E}}  dU\,g(U)\left(1+\beta^{|F|}\,\sum_{H\in \{1,i,j,k\}^V}\prod_{AB\in E} \mathrm{Re}^2\,(H(A)^*U(AB)H(B))\right),
  %%\,\prod_{AB}dU(AB).
  %%%%%%%%%%
  %\hspace{-0.8cm}\int_{G^{E}} dUf(U)\prod_{ABC}\left(1+\beta
%\,\mathrm{Re}\,(U(AB)U(BC)U(CA))
%\right)%\,\prod_{AB}dU(AB)
  %&=
  %\int_{G^{E}}  dUf(U)\left(1+\beta^{|F|}\,\sum_{H\in S^V}\prod_{AB}\mathrm{Re}^2\,(H(A)^*U(AB)H(B))\right)
  %%\,\prod_{AB}dU(AB).
  %\end{align*}
  \end{multline*}
where we set $U(AB),\sigma(A),H(A):=1$ for each boundary edge $AB$ or boundary vertex $A$.
%%%%
%For each integrable $f\in L_1(LinearizedGaugeTheory$(U(1),\beta,\Omega,\sigma))$ invariant under all transformations $\theta()\!\mapsto\!\pi\!\pm\!\theta_n\!\!\!\!\!\!\pmod{2\pi}$ we have}
\end{proposition}
\vspace{-0.1cm}

\begin{figure}[hbt]
\centering
\includegraphics[width=2.7cm]{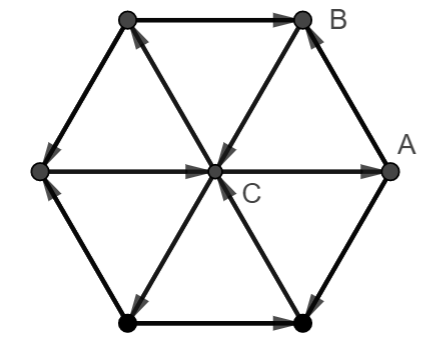}
\caption{A hexagon triangulated by $6$ regular triangles; see Example~\ref{ex-small-hexagon}}
\label{fig2}
\end{figure}

The following more abstract result expresses the relations among the models in detail; the required definitions are given in the next sections.

\begin{proposition} \label{th-split-intro}
For each polygon $\Omega$ triangulated by regular triangles and maps $\sigma$, $h$ of the set of boundary vertices to $\{+1,-1\}$, $\{1,i,j,k\}$ %$\{\mathrm{e}_1,\dots,\mathrm{e}_8\}$
respectively, the following probability spaces are isomorphic:
\begin{align*}
	\hspace{-0.8cm}
	&\mathrm{LinearizedGaugeTheory}(U(1),\infty,\Omega,\sqrt{\sigma})
\ltimes \mathrm{Percolation}(\Omega,\sigma)
	&\cong
	&\mathrm{RandomMediumIsing}(\Omega,\sigma)\\
	&&\cong
	&\mathrm{LinearizedHiggs}(U(1),\{1,i\},0,\infty,\Omega,\sqrt{\sigma}),\\
	&\mathrm{LinearizedGaugeTheory}(SU(2),\infty,\Omega,h)
\ltimes \mathrm{Potts}(4,\Omega,\infty,h)
	&\cong
	&\mathrm{LinearizedHiggs}(SU(2),\{1,i,j,k\},0,\infty,\Omega,h).
	%%%
	%\\
	%&\mathrm{LinearizedGaugeTheory}(S^7,\infty,\Omega,{h})
	%&\cong
	%&\mathrm{LinearizedHiggs}(S^7,\{\mathrm{e}_1,\dots,\mathrm{e}_8\},0,\infty,\Omega,h')
	%/\mathrm{Potts}(8,\Omega,\infty,h).
\end{align*}
%%%%%%%%%%%%%%%%
%\begin{align*}
%\hspace{-0.8cm}
%&\mathrm{LinearizedGaugeTheory}(U(1),\infty,\Omega,\sqrt{\sigma})
%&\cong
%&\mathrm{RandomMediumIsing}(\Omega,\sigma)/\mathrm{Percolation}(\Omega,\sigma)\\
%&&\cong
%&\mathrm{LinearizedHiggs}(U(1),\{1,i\},0,\infty,\Omega,\sqrt{\sigma})
%/\mathrm{Percolation}(\Omega,\sigma),\\
%&\mathrm{LinearizedGaugeTheory}(SU(2),\infty,\Omega,h)
%&\cong
%&\mathrm{LinearizedHiggs}(SU(2),\{1,i,j,k\},0,\infty,\Omega,h)
%/\mathrm{Potts}(4,\Omega,\infty,h).
%\end{align*}
\end{proposition}

%\begin{corollary}
%  Conformal invariance of the crossing probability in $\mathrm{RandomMediaIsing}(\Omega,\sigma)$.
%\end{corollary}

%Let us state another straightforward corollary.
%
%\begin{corollary} \mscomm{(To be stated properly)}
%  Conformal invariance of the expectation of the random variable $\mathrm{E}(f\circ p_2|p_1)(U)$ in $\mathrm{LinearizedGaugeTheory}(U(1),\infty,\Omega,\sqrt\sigma))$.
%\end{corollary}

%\section{Preliminaries}

\section{Linearization of Abelian lattice gauge theory}

%Let us give precise definitions of the models involved. %No background in physics is required.

Let us give a zero-knowledge introduction to gauge theory; cf.~\cite{Maldacena-16, Creuz-70}.

We start with an informal toy model. Let $7$ cities be connected by oriented roads as in Figure~\ref{fig2}. Each city has its own type of goods in an unlimited quantity. E.g., city $A$ has apples, city $B$ has bananas, and city $C$ has coconuts. On each oriented road $AB$, an exchange rate $U(AB)>0$ is fixed. E.g.,  travelling from $A$ to $B$, one gets $2$ banana for an apple.

A cunning citizen can travel and exchange along a triangle $ABC$ to multiply his initial amount of goods by a factor of $U(AB)U(BC)U(CA)$.  The total speculation profit is measured by the quantity
$$
S(U):=\sum_{\text{all faces } ABC}\log^2(U(AB)U(BC)U(CA)).
%\mathcal{S}_{\textrm{Manton}}[U]:=\sum_{uvwx}\ln^2(U(uv)U(vw)U(wx)U(xu)
%\text{ or }
%\mathcal{S}_{\textrm{Wilson}}[U]:=\sum_{uvwx}(1-\cos\ln(U(uv)U(vw)U(wx)U(xu)),
$$
Here $\log^2(x)$ is chosen as a function vanishing at $x=1$ and positive for $x\ne 1$. A useful observation: for a city $C$, one can change the units of measurements, e.g., exchange dozens of coconuts instead of single ones. Such \emph{gauge transformation} multiplies the rates for all the roads starting at $C$ and divides the rates for all the roads ending at $C$ by the same value but preserves ${S}(U)$.

One can turn this economic model into a statistical-physics one by fixing the boundary rates and making the interior ones random with the probability density $P(U)$ exponentially decreasing with the growth of the speculation profit ${S}(U)$. The resulting model describes quantum electromagnetic field.

%The collection of rates $U$ minimizing the quantity $\mathcal{S}[U]$ among all the rates with the same values at the boundary is a \emph{classical gauge field}.

%Let us discuss the crucial properties of the model.

In gauge theory used in particle physics, the rates become complex numbers, quaternions etc.

\begin{definition} \label{def-lattice-gauge} Let $G$ be a submanifold of one of the sets
$$
U(1):=\{z\in\mathbb{C}:|z|=1\},
\quad
SU(2)=\{z\in\mathbb{H}:|z|=1\},
\quad\text{ or }\quad
S^7=\{z\in\mathbb{O}:|z|=1\},
$$
of norm-one complex numbers, quaternions, and octonions respectively. (To get the idea of what follows, it is suggested to start with
%first restrict to
the case when $G=U(1)$ everywhere, known as \emph{compact Abelian gauge theory}.) Let $\Omega$ be a polygon triangulated by  regular triangles %(see Figure~\ref{fig2}).
Let $F$, $E$, $V$ be the sets of faces, nonboundary edges, nonboundary vertices in the triangulation respectively. Assume that all the edges are parallel to the $3$ cubic roots of $1$.
%the vectors $1, \zeta, \zeta^2\in\mathbb{C}$, where $\zeta$ is a cubic root of $1$.
Orient the edges in the directions of the roots.
Denote by $AB$ the edge oriented from $A$ to $B$.
Fix $0\le\beta\le1$ and an element $u(AB)\in G$ for each boundary edge $AB$.

%For each $U(1)$ valued function $U$ on the set of edges denote
Let $G^{E}$ be the set of $G$-valued functions $U$ on the set of edges which are equal to $u$ on the boundary. Define LatticeGaugeTheory$(G,\beta,\Omega,u)$ to be the probability space $G^{E}$ with the probability density
$$
P(U)=\frac{1}{Z}\exp\left(\beta\sum_{ABC\in F}%(
\mathrm{Re}\,(U(AB)U(BC)U(CA))%-1)
\right).
$$
Here the sum is over all the triangles $ABC$ %(the vertices are listed in an arbitrary order)
(the vertices are listed in an order compatible with the orientation of the edges) and
$Z\in\mathbb{R}$ is chosen so that the total probability $\int_{G^{E}}P(U)dU=1$, where $dU$ denotes the product of Lebesque measures (or counting measures, if $G$ is finite, i.e., a $0$-dimensional submanifold). The constants $Z',Z'',\dots$ below are chosen analogously.

Define PreLinearizedGaugeTheory$(G,\beta,\Omega,u)$ to be the %probability
space $G^{E}$ with the probability density
$$
P'(U)=\frac{1}{Z'}\prod_{ABC\in F}\left(1+\beta
\,\mathrm{Re}\,(U(AB)U(BC)U(CA))
\right).
%\qquad\text{and}\qquad
%P(U)=\frac{1}{Z}\prod_{ABC}
%\mathrm{Re}\,(U(AB)U(BC)U(CA)).
$$
%
%For an edge $AB$ oriented from $A$ to $B$ denote $U(BA):=U(AB)^{-1}$.

In the case when $G$ is a group, given a $G$-valued function $g$ on the set of nonboundary vertices, define the \emph{gauge transformation} $\phi_g\colon G^{E}\to G^{E}$ by the formula $[\phi_gU](AB)=g(A)U(AB)g(B)^*$.
%Hereafter notation `$U(AB)$' always means that the edge $AB$ is oriented from $A$ to $B$.
\end{definition}

\begin{remark} Definition~\ref{def-lattice-gauge} is applicable even if $G$ is not a group, e.g., $G\!=\!S^7$ or $G\!=\! S^2\!=\! SU(2)\cap\{\mathrm{Re}\,z\!=\!0\}$. The expression for $P(U)$ is well-defined for $G\subset S^7$ because $\mathrm{Re}((xy)z)=\mathrm{Re}(x(yz))$ for any %octonions
$x, y, z\in\mathbb{O}$ \cite[Eq.~(6.21)]{Ootsuka-etal-05}. Thus octonion gauge theory on a triangular lattice is defined easier than the continuum one \cite{Ootsuka-etal-05}. The price is that gauge transformations are not well-defined when $G$ is not a group. Also we conjecture that Proposition~\ref{prop-linearization-intro} does not remain true literally for $SU(2)$ replaced by $S^7$ or $S^2$;  see Remark~\ref{ex-octonion} and~Proposition~\ref{th-split-S2}.
The case $G=SU(2)$ is often called \emph{Yang-Mills theory}.
%$P(U)$ is \emph{not} gauge-invariant for $G=S^7$. If $G$ is not closed under multiplication, then gauge transformation are even not well-defined.
%domain of $U$ is \emph{not} gauge-invariant (e.g., for $G=S^2$).

Definition~\ref{def-lattice-gauge} is applicable if $G$ is finite, %(i.e., $0$-dimensional submanifold),
e.g., $G\!=\!\{+1,-1\}$. Then the \emph{probability density} is understood as the function $P(U)$ related to the probability measure $\mu$ via $\mu(A)\!\!=\!\!\!\int_A \! P(U)dU\!:=\!\!\!\sum_{U\in A}\!P(U)$.

Definition~\ref{def-lattice-gauge} immediately generalizes to an arbitrary Lie group $G$, e.g., $G\!=\!SU(3)$. It is interesting to find an analogue of Proposition~\ref{prop-linearization-intro} in this setup. Generalization to higher-dimensional lattices is also easy \cite[\S7]{Creuz-70}.
\end{remark}

%\begin{remark}
%  Beware that edges are \emph{not} oriented. In particular, $U(AB)=U(BA)$ for each edge $AB$, and $U(AB)U(BC)U(CA)$ does not depend on the ordering of the vertices of a triangle $ABC$. This is different from the commonly used notation, when $U$ is defined on oriented edges. `Our' $U(AB)$ equals the `usual' $U(AB)$, when all the edges are oriented in the directions of the $3$ cubic roots of $1$.
%\end{remark}

\begin{table}
  \caption{Relation between lattice gauge theory and the Ising model; see Example~\ref{ex-small-hexagon}}
  %Graphical presentation of terms in the proof of Example~\ref{ex-small-hexagon}
  \label{tab-rel}

  \smallskip
  \centering
  \begin{tabular}{|c|c|c|c|}
    \hline
    $\sin\theta_n\cos \theta_{n+1}$ &
    $\cos\theta_n\sin \theta_{n+1}$ &
    $\sin\theta_n\sin \theta_{n+1}$ &
    $\cos\theta_n\cos \theta_{n+1}$      \\
    \hline
    & & & \\[-0.3cm]
    \includegraphics[width=0.2\textwidth]{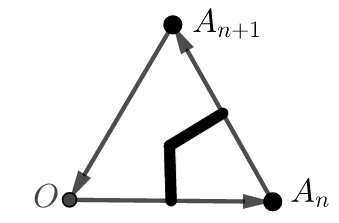} &
    \includegraphics[width=0.2\textwidth]{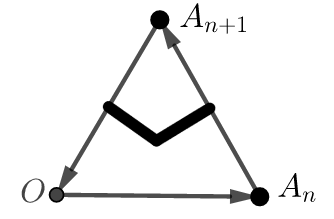} &
    \includegraphics[width=0.2\textwidth]{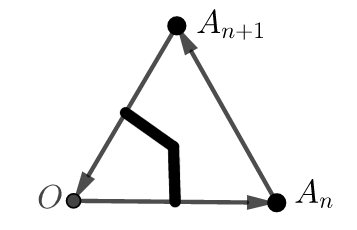} &
    \includegraphics[width=0.2\textwidth]{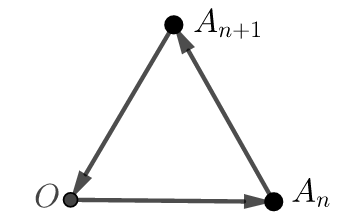}  \\
    \hline
  \end{tabular}
%\end{table}
%\begin{figure}[htbp]
%\begin{center}
%\hspace{-0.7cm}
\smallskip

\begin{tabular}{|c|c|c|}
\hline
{\scriptsize
$\cos^2\theta_1\sin^2\theta_2\cos^2\theta_3\cos^2\theta_4\sin^2\theta_5\cos^2\theta_6$ }
&
{\scriptsize
$\sin^2\theta_1\cos^2\theta_2\sin^2\theta_3\sin^2\theta_4\cos^2\theta_5\sin^2\theta_6$ }
&
{\scriptsize
$\sin\theta_1\cos^2\theta_2\sin^2\theta_3
\sin^2\theta_4\cos\theta_5\sin\theta_5\cos^2\theta_6\cos\theta_1$
}
\\
\hline
 & & \\[-0.3cm]
\includegraphics[width=5cm]{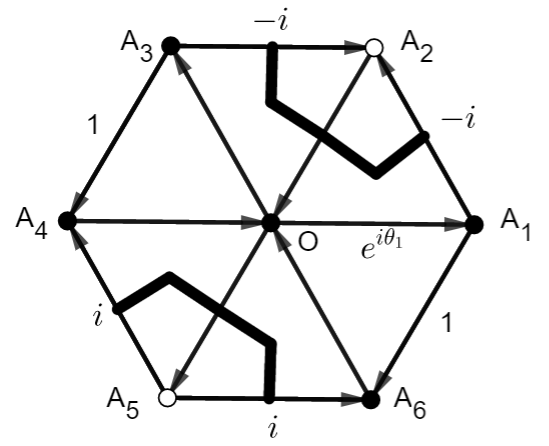} &
\includegraphics[width=5cm]{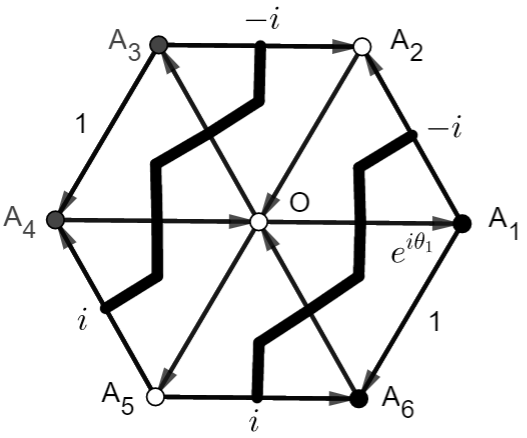} &
\includegraphics[width=5cm]{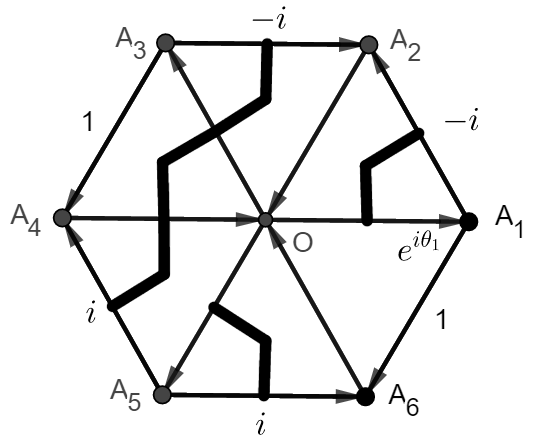} \\
\hline
\end{tabular}
\end{table}
%\end{center}
%\caption{Relation between lattice gauge theory and percolation}
%\label{fig2}
%\end{figure}

In the %constructed
model, one usually studies expectations of certain random variables. Those are easy to compute (on a $2$-dimensional lattice) using gauge transformations, if a random variable is invariant under the transformations, and nontrivial otherwise. Proposition~\ref{prop-linearization-intro} asserts that for a random variable having just coordinate-reflection symmetries, the computation still simplifies much: one may drop most terms in the expansion of $P'(U)$ in $\beta$. We illustrate this in the simplest nontrivial particular case.

\begin{example} \label{ex-small-hexagon}
Let $G=U(1)$ be identified with $[0,2\pi)$ by the exponential mapping.
Let $\Omega$ be a regular
hexagon $A_1A_2A_3A_4A_5A_6$ triangulated by $6$ regular triangles with a common vertex $O$; $A_7:=A_1$.
Set
$$
u(A_1A_{2})\!=\!u(A_3A_{2})\!=\!-i,\quad
u(A_3A_{4})\!=\!u(A_1A_{6})\!=\!1,\quad
u(A_5A_{4})\!=\!u(A_5A_{6})\!=\!i\quad\text{(see Table~\ref{tab-rel}, bottom)}.
$$
In PreLinearizedGaugeTheory$(U(1),\beta,\Omega,u)$, we have $Z'=(2\pi)^{6}+2\pi^{6}\beta^6$ and for each integrable function $f\colon [0,2\pi]^6\!\to\! \mathbb{R}$ invariant under the transformations $\theta_n\!\mapsto\!\pi\!\pm\!\theta_n\!\pmod{2\pi}$
for $n\!=\!1,\dots,6$ we have
\begin{align}
%\intertext{Analogously we get}
%which is even and $\pi$-periodic in each variable,
%\begin{multline*}
\mathrm{E}f&=
%\frac{1}{Z'}
\int_{[0,2\pi]^6}f(\theta_1,\dots,\theta_6)
P''_\beta(\theta_1,\dots,\theta_6)\,d\theta_1\dots d\theta_6,
%\end{multline*}
\label{eq-E}
\intertext{where}
P''_\beta(\theta_1,\dots,\theta_6)&:=\frac{1}{Z'} \left(1+\beta^6
\cos^2\theta_1\sin^2\theta_2\cos^2\theta_3
\cos^2\theta_4\sin^2\theta_5\cos^2\theta_6\right.\notag
\\
&\hspace{0.09\linewidth}
+\left.\beta^6\sin^2\theta_1\cos^2\theta_2\sin^2\theta_3
\sin^2\theta_4\cos^2\theta_5\sin^2\theta_6\right).\notag
%\\[-1.7cm]
\notag
\end{align}
%%%
%Then for SimplifiedGaugeTheory$(U(1),\beta,\Omega,u)$, the number $Z'$ is given by equation~\eqref{eq-Z} below, and for each integrable function $f\colon [0,2\pi]^6\!\to\! \mathbb{R}$ invariant under the transformations $\theta_n\!\mapsto\!\pi\!\pm\!\theta_n\!\pmod{2\pi}$for each $n=1,\dots,6$, the expectation $\mathrm{E}f$ is given  by equation~\eqref{eq-E} (which we are going to derive now).
\end{example}

\begin{remark} %This is useful only if $f$ is not gauge invariant; in particular, there are easier ways to find~$Z'$.
Equation~\eqref{eq-E} is only useful if $f$ is \emph{not} gauge invariant; otherwise there is a much simpler way to compute the expectation. Likewise, there is a simpler way to compute~$Z'$ than described below; but we discuss the method which automatically gives~\eqref{eq-E} as well.
\end{remark}
%Let us compute $Z$ for SimplifiedGaugeTheory$(U(1),\infty,\Omega,u)$.

%Let us expand $Z$ in powers of $\beta$ and compute the first two nonzero terms.
%% The first term is obviously $(2\pi)^{6}$; let us find the next one.
\begin{proof}[Proof of Example~\ref{ex-small-hexagon}]
Given $U\in U(1)^{E}$, denote $\exp(i\theta_n):=\begin{cases}U(OA_n),&\text{for $n=1,3,5$};\\ U(A_nO),&\text{for $n=2,4,6$.}\end{cases}$
%and each $n=1,\dots,6$ denote $U(OA_n)=:\exp(i\theta_n)$. %, $d\theta:=d\theta_1\dots d\theta_6$, $C:=[0,2\pi]^6$.
Then
\begin{align*}
%%%%%%%%%%%%%
%Z&=\int_{U(1)^6}dU\,\prod_{n=1}^6
%\mathrm{Re}\,(U(OA_n)U(A_nA_{n+1})U(A_{n+1}O))\\
%&=\int_{C}d\theta\,
%\sin(\theta_1+\theta_2)\sin(\theta_2+\theta_3)\cos(\theta_3+\theta_4)
%\sin(\theta_4+\theta_5)\sin(\theta_5+\theta_6)\cos(\theta_6+\theta_1)\\
%&=
%\int_{C}d\theta\,
%(\sin\theta_1\cos\theta_2\!+\!\cos\theta_1\sin\theta_2)
%(\sin\theta_2\cos\theta_3\!+\!\cos\theta_2\sin\theta_3)
%(\cos\theta_3\cos\theta_4\!-\!\sin\theta_3\sin\theta_4)
%\times\\
%&{}\hspace{0.07\linewidth}\times
%(\sin\theta_4\cos\theta_5\!+\!\cos\theta_4\sin\theta_5)
%(\sin\theta_5\cos\theta_6\!+\!\cos\theta_5\sin\theta_6)
%(\cos\theta_6\cos\theta_1\!-\!\sin\theta_6\sin\theta_1)
%\\
%&=
%\int_{C}d\theta\,
%\left(\cos^2\theta_1\sin^2\theta_2\cos^2\theta_3\cos^2\theta_4\sin^2\theta_5\cos^2\theta_6
%+\sin^2\theta_1\cos^2\theta_2\sin^2\theta_3\sin^2\theta_4\cos^2\theta_5\sin^2\theta_6\right)\\
%&=
%2\pi^{6}.
%%%%%%%%%%%%
%Z&=\int_{U(1)^6}dU\,\exp\left(\beta\sum_{n}
%\mathrm{Re}\,(U(OA_n)U(A_nA_{n+1})U(A_{n+1}O))\right)\\
%&=\int_{C}d\theta\,\exp\left(\beta
%(\sin(\theta_1+\theta_2)+\sin(\theta_2+\theta_3)+
%\cos(\theta_3+\theta_4)+\sin(\theta_4+\theta_5)+\sin(\theta_5+\theta_6)+
%\cos(\theta_6+\theta_1))
%\right)\\
Z'&=
%\int_{U(1)^6}dU\,\prod_{n=1}^6\left(1+\beta\,
%\mathrm{Re}\,(U(OA_n)U(A_nA_{n+1})U(A_{n+1}O))\right)\notag\\
\int_{U(1)^6}dU\,\left(1+\beta\,
\mathrm{Re}\,(U(OA_1)U(A_1A_{2})U(A_{2}O))\right)
\cdots
\left(1+\beta\,
\mathrm{Re}\,(U(A_6O)U(OA_{1})U(A_1A_{6}))\right)
\notag\\
&=
\int_{[0,2\pi]^6}d\theta_1\dots d\theta_6\,
(1\!+\!\beta\sin(\theta_1\!+\!\theta_2))
(1\!+\!\beta\sin(\theta_2\!+\!\theta_3))
(1\!+\!\beta\cos(\theta_3\!+\!\theta_4))\times\notag\\
&{}\hspace{0.19\linewidth}
%&{}\hspace{0.066\linewidth}
\times(1\!-\!\beta\sin(\theta_4\!+\!\theta_5))
(1\!-\!\beta\sin(\theta_5\!+\!\theta_6))
(1\!+\!\beta\cos(\theta_6\!+\!\theta_1))\notag\\
&=
\int_{[0,2\pi]^6}d\theta_1\dots d\theta_6\,
(1\!+\!\beta\sin\theta_1\cos\theta_2\!+\!\beta\cos\theta_1\sin\theta_2)
(1\!+\!\beta\sin\theta_2\cos\theta_3\!+\!\beta\cos\theta_2\sin\theta_3)
\times\notag\\
&{}\hspace{0.19\linewidth}
%&{}\hspace{0.066\linewidth}
\times(1\!+\!\beta\cos\theta_3\cos\theta_4\!-\!\beta\sin\theta_3\sin\theta_4)
(1\!-\!\beta\sin\theta_4\cos\theta_5\!-\!\beta\cos\theta_4\sin\theta_5)
\times\notag\\
&{}\hspace{0.19\linewidth}
%&{}\hspace{0.066\linewidth}
\times(1\!-\!\beta\sin\theta_5\cos\theta_6\!-\!\beta\cos\theta_5\sin\theta_6)
(1\!+\!\beta\cos\theta_6\cos\theta_1\!-\!\beta\sin\theta_6\sin\theta_1)
\notag\\
&\overset{(*)}{=}\int_{[0,2\pi]^6}d\theta_1\dots d\theta_6\,
\left(1+\beta^6\cos^2\theta_1\sin^2\theta_2\cos^2\theta_3
\cos^2\theta_4\sin^2\theta_5\cos^2\theta_6\notag\right.
\notag\\
&\hspace{0.24\linewidth}
+\left.\beta^6\sin^2\theta_1\cos^2\theta_2\sin^2\theta_3
\sin^2\theta_4\cos^2\theta_5\sin^2\theta_6
\right)
\notag\\
%%%%%%%%%%%%%%%%
%%%&+
%%%\beta^{12}\int_{C}d\theta\,\left(
%%%\cos^4\theta_1\sin^4\theta_2\cos^4\theta_3\cos^4\theta_4\sin^4\theta_5\cos^4\theta_6
%%%+\sin^4\theta_1\cos^4\theta_2\sin^4\theta_3\sin^4\theta_4\cos^4\theta_5\sin^4\theta_6\right)\\
%&+O(\beta^{8})\\
%%%%%%%%%%%%%%%%
&=
(2\pi)^{6}+2\pi^{6}\beta^6.\label{eq-Z}
\end{align*}
%Here `\dots' denote higher-order terms in $\beta$.
Let us explain key equality (*).
We expand the product in the left-hand side of~(*)% before the equality
, and to each resulting term assign a ``Feynman diagram'' as follows; see Table~\ref{tab-rel}.

First consider a term obtained by taking a summand distinct from $1$ in each factor in the left-hand side of~(*). If the taken summand has form $\sin\theta_n\cos\theta_m$, where $m=n\pm 1$, then draw a segment from the center of the triangle $OA_mA_n$ to the midpoint of $OA_n$ (see Table~\ref{tab-rel}, top-left). If it has form $\sin\theta_n\sin\theta_{n+1}$, draw $2$ segments from the center of $OA_nA_{n+1}$ to the midpoints of $OA_n$ and $OA_{n+1}$ (see Table~\ref{tab-rel}, top-middle). For the summands of the form $ \cos\theta_n\cos\theta_{n+1}$ draw nothing (see Table~\ref{tab-rel}, top-right). Also draw $4$ segments from the centers of $OA_nA_{n+1}$ to the midpoints of $A_nA_{n+1}$ for $n=1,2,4,5$. We get an even number of segments in each triangle, thus all the segments form several broken lines with the endpoints at the edge midpoints.

Now observe that if at least one of the endpoints does not belong to the boundary of $\Omega$ (see Table~\ref{tab-rel}, bottom-right), then the term has a factor $\sin\theta_n\cos\theta_n$ vanishing after integration over $\theta_n$. The only two terms giving a nonzero contribution to the integral are the ones in Table~\ref{tab-rel}, bottom-left and bottom-middle.
They are in bijection with the $2$-colorings of vertices such that $A_n$ is black for $n=1,3,4,6$ and white for $n=2,5$. (Our broken lines separate vertices of different color.)

Finally, a similar %`diagram'
argument shows that taking at least one summand $1$ leads to just one term contributing to the integral. We are thus left with only $3$ contributing terms, which proves~(*). Equation~\eqref{eq-E}
is proven analogously.
%Denote the remaining terms with nonzero contribution by
%%%
%\begin{align}
%P''(\theta_1,\dots,\theta_6)&:=1+\beta^6
%\cos^2\theta_1\sin^2\theta_2\cos^2\theta_3
%\cos^2\theta_4\sin^2\theta_5\cos^2\theta_6\notag
%\\
%&\hspace{0.05\linewidth}
%+\beta^6\sin^2\theta_1\cos^2\theta_2\sin^2\theta_3
%\sin^2\theta_4\cos^2\theta_5\sin^2\theta_6.\notag
%\intertext{Then }
%Z'&=\int_{[0,2\pi]^6}
%P''(\theta_1,\dots,\theta_6)\,d\theta_1\dots d\theta_6
%=(2\pi)^{6}+2\pi^{6}\beta^6. \label{eq-Z}
\end{proof}

This example suggests to replace $P'$ by $P''_\beta$ in Definition~\ref{def-lattice-gauge}. In particular, this extends the definition to $\beta>1$.
Beware that $P''_\beta(\theta_1,\dots,\theta_6)\ne P'(e^{i\theta_1},\dots,e^{i\theta_6})$ in general.

\begin{definition} \label{def-linearized-gauge}
Assume that $G=U(1)$. Fix $\beta>0$ and a sign $\sigma(A)=\pm1$ for each boundary vertex $A$. For each boundary edge $AB$ set $u(AB)=1$, if $\sigma(A)=\sigma(B)$, and $u(AB)=i\sigma(B)$, if $\sigma(A)\ne\sigma(B)$.

%Let $[0,2\pi]^{E}$ be the set of $[0,2\pi]$-valued functions $\Theta$ on the set of edges equal to $\theta$ at the boundary.
Let $\{+1,-1\}^{V}$ be the set of $\pm 1$-valued functions $\Sigma$ on vertices, equal to $\sigma$ on the boundary. Denote
$$
P(U,\Sigma):=
\prod_{AB:\Sigma(A)=\Sigma(B)}   \cos^2\theta(AB)
\prod_{AB:\Sigma(A)\ne\Sigma(B)} \sin^2\theta(AB),
$$
where $U(AB)=:\cos\theta(AB)+i\sin\theta(AB)$.
%where $\exp(i\theta(AB)):=U(AB)$.
%where the products are over edges $AB$.
Now define LinearizedGaugeTheory$(U(1),\beta,\Omega,\sqrt{\sigma})$
and LinearizedGaugeTheory$(U(1),\infty,\Omega,\sqrt{\sigma})$
to be the space $U(1)^{E}$ %$\times \{+1,-1\}^{V}$
with the probability densities %respectively
$$
P''_\beta(U)=\frac{1}{Z''_\beta}\left(
1+\beta^{|F|}\sum_{\Sigma\in\{+1,-1\}^{V}}P(U,\Sigma)\right)
\quad\text{and}\quad
P''_\infty(U)=\frac{1}{Z''_\infty}\sum_{\Sigma\in\{+1,-1\}^{V}}P(U,\Sigma).
$$
\end{definition}

\begin{remark} %Linearized gauge theory approximates lattice gauge theory in the same manner as the \emph{loop} $O(n)$-model approximates the \emph{spin} $O(n)$-model \cite[\S1]{Duminil-Copin-etal-14}.
The origin of the notation $\sqrt\sigma$ and the formula for $u(AB)$ are going to become clear in the next section.
\end{remark}

\begin{example}\label{ex-f} Consider the hexagon $\Omega$ from Example~\ref{ex-small-hexagon}. Let $\sigma(A_n)$ be $+1$ for $n = 1,3,4,6$ and be $-1$ for $n=2,5$. Then for
LinearizedGaugeTheory$(U(1),\beta,\Omega,\sqrt\sigma)$,
the function $P''_\beta$ is given by the expression from that example.
The crossing probability for site percolation on $\Omega$ with the boundary condition~$\sigma$ (which  itself equals $1/2$) equals the expectation of the random variable
$$
f(\theta_1,\dots,\theta_6)=\frac{\cos^2\theta_1\sin^2\theta_2\cos^2\theta_3
\cos^2\theta_4\sin^2\theta_5\cos^2\theta_6}
{\cos^2\theta_1\sin^2\theta_2\cos^2\theta_3
\cos^2\theta_4\sin^2\theta_5\cos^2\theta_6
+\sin^2\theta_1\cos^2\theta_2\sin^2\theta_3
\sin^2\theta_4\cos^2\theta_5\sin^2\theta_6}
$$
in LinearizedGaugeTheory$(U(1),\infty,\Omega,\sqrt\sigma)$.
The random variable $f$ is \emph{not} gauge invariant.
\end{example}

In the random variable $f$, one recognizes the crossing probability in the Ising model with edge weights $\cot\theta_1,\dots,\cot\theta_6$. Such \emph{cotan weights} appear naturally, e.g., on isoradial graphs \cite{Chelkak-Smirnov-08}.
This observation is summarized in the following definition and %easy
Proposition~\ref{th-split-intro}.

\begin{definition} \label{def-random-media-Ising}
%First equip the space $U(1)^{E}\times \{+1,-1\}^{V}$ with the product of the Lebesque measure and the counting measure. Then equip it with a probability measure using the probability density $P(U,\Sigma)/Z$ for suitable $Z\in\mathbb{R}$.
Define RandomMediumIsing$(\Omega,\sigma)$ to be the disconnected manifold $U(1)^{E}\times \{+1,-1\}^{V}$ with the probability density $P(U,\Sigma)/Z''_\infty$. Define Percolation$(\Omega,\sigma)$ to be the space $\{+1,-1\}^{V}$ with the counting measure divided by $2^{{|V|}}$.
%
%Define the obvious projections
%$$
%\text{Signs}\colon [0,2\pi]^{E}\times \{+1,-1\}^{V}\to \{+1,-1\}^{V}\text{ and }
%\mathrm{Weights}\colon [0,2\pi]^{E}\times \{+1,-1\}^{V}\to
%[0,2\pi]^{E}.
%$$
\end{definition}

\begin{definition}\label{def-quotient}
%  Let $(M,\Sigma,\mu)$ be a probability space such that $M=X\times Y$ and $\Sigma=\Sigma_X\otimes \Sigma_Y$. Denote by $M/Y$ the probability space $(X,\Sigma_X,\nu)$ with the measure given by $\nu(A)=\mu(A\times Y)$ for each $A\in\Sigma_X$.
A \emph{semi-direct product} of probability spaces $\mathrm{X}=(X,\Sigma_X,\mu_X)$ and $\mathrm{Y}=(Y,\Sigma_Y,\mu_Y)$ is a probability space $\mathrm{Z}=(X\times Y,\Sigma_X\otimes \Sigma_Y,\mu_Z)$ with any measure $\mu_Z$ such that  $\mu_Z(A\times Y)=\mu_X(A)$ and  $\mu_Z(X\times B)=\mu_Y(B)$ for each $A\in\Sigma_X$, $B\in\Sigma_Y$. Notation: $\mathrm{Z}=\mathrm{X}\ltimes\mathrm{Y}$ or $\mathrm{Z}=\mathrm{X}\rtimes\mathrm{Y}$.
\end{definition}

\begin{remark} %This definition is symmetric with respect to $\mathrm{X}$ and $\mathrm{Y}$ in spite of nonsymmetric notation $\ltimes$.
	The notion is closely related but different from the \emph{semi-direct product of measures}. % \cite{Bolthausen}.
\end{remark}

%\begin{proposition} \label{th-U(1)-split}
%For each triangulated polygon $\Omega$ and assignment $\sigma$ of $\pm1$ to the boundary vertices,
%$$\mathrm{RandomMediaIsing}(\Omega,\sigma)=
%\mathrm{LinearizedGaugeTheory}(U(1),\infty,\Omega,\sigma))\times
%\mathrm{Percolation}(\Omega,\sigma)$$
%as probability spaces.
%\end{proposition}

\begin{proof}[Proof of the 1st isomorphism in Proposition~\ref{th-split-intro}] This is straightforward:
%follows directly from the definitions.
Let $\mu$ and $\nu$ be the probability measures in $\mathrm{LinearizedGaugeTheory}(U(1),\infty,\Omega,\sqrt\sigma)$ and
$\mathrm{RandomMediumIsing}(\Omega,\sigma)$ respectively, and $\lambda$ be the Lebesque measure on $U(1)^E\times \{+1,-1\}^{V}$.
Then for each measurable $A\subset U(1)^E$
$$
\nu(A)=\int_A P''_\infty(U)\,dU
=\frac{1}{Z''_\infty}\int_A\,\sum_{\Sigma\in\{+1,-1\}^{V}}P(U,\Sigma)\,dU
=\frac{1}{Z''_\infty}\int_{A\times \{+1,-1\}^{V}}P(U,\Sigma)\,d\lambda
=\mu(A\times\{+1,-1\}^{V}).
$$
%\end{proof}
%
%\begin{remark}
Since $\int_0^{2\pi}\cos^2\theta\,d\theta=\int_0^{2\pi}\sin^2\theta\,d\theta$, it follows that $\int_{U(1)^{E}}P(U,\Sigma)\,dU$ is the same
%=\pi^{{|E|}}$ wrong!!!
for each $\Sigma$,
%\in\{+1,-1\}^{V}$,
thus
$$
\mu(U(1)^E\times B)=\frac{1}{Z''_\infty}\int_{U(1)^E\times B}P(U,\Sigma)\,d\lambda=\frac{|B|}{2^{|V|}}
$$
for each  $B\subset\{+1,-1\}^{V}$. By Definition~\ref{def-quotient} we get the 1st isomorphism in Proposition~\ref{th-split-intro}.
%$$
%\mathrm{Percolation}(\Omega,\sigma)\cong
%\mathrm{RandomMediumIsing}(\Omega,\sigma)/\mathrm{LinearizedGaugeTheory}(U(1),\infty,\Omega,\sqrt\sigma).
%$$
%\end{remark}
\end{proof}

\begin{corollary}\label{cor-assignment} For each random variable $f$ on $\mathrm{Percolation}(\Omega,\sigma)$ the random variable
$$
\mathrm{E}(f\circ p_2|p_1)(U)
:=\frac{\sum_{\Sigma\in\{+1,-1\}^{V}}f(\Sigma)P(U,\Sigma)}
{\sum_{\Sigma\in\{+1,-1\}^{V}}P(U,\Sigma)}
$$
on $\mathrm{LinearizedGaugeTheory}(U(1),\infty,\Omega,\sqrt\sigma)$ has the same expectation (but \emph{not} variance).
\end{corollary}

%As a straightforward corollary, Smirnov's Theorem implies that the crossing probability in the random-media Ising model is conformally invariant, and we get a conformally invariant observable in linearized gauge theory (see Corollary~\ref{cor-assignment} and~Remark~\ref{rem-conformal}).

\begin{remark} \label{rem-conformal}
In particular, if $f$ is the indicator function of the crossing event in $\Omega$, then by the Cardy's formula proved by S.~Smirnov \cite{Smirnov-01,KS-20}, $\mathrm{E}(f\circ p_2|p_1)$ is an observable in linearized gauge theory, which is conformally invariant in the continuum limit. (We consider the notation $\mathrm{E}(f\circ p_2|p_1)$ as indecomposable to avoid discussion of its ingredients.)
\end{remark}

%\begin{remark} The same does \emph{not} hold for the variance.
%The maps preserve $L_1$ but not $L_2$ norm. \mscomm{Recheck!!!}
%\end{remark}

\begin{remark} This result extends to \emph{Smirnov's parafermionic observable}. Let us informally sketch that. In Definition~\ref{def-lattice-gauge}, fix a nonboundary edge $AB$.
Let the function $U$ be \emph{two}-valued at the edge $AB$; let the two values $U_+(AB)$ and $U_-(AB)$ be related by $U_-(AB)=iU_+(AB)$. In the formulas for $P(U)$ and $P'(U)$, replace $U(AB)$ by $U_+(AB)$ in the summand corresponding to the triangle $ABC$ bordering upon $AB$ from the right, and by $U_-(AB)$ --- for the one from the left. In Definition~\ref{def-linearized-gauge}, replace $\{+1,-1\}^V$ by the set of loop and arc configurations with one of the endpoints at the midpoint of $AB$. In the formula for $P(U,\Sigma)$, take $\cos^2\theta(AB)$, if an arc arrives at the midpoint of $AB$ \emph{from the right}, and $\sin^2\theta(AB)$ otherwise; here $\cos\theta(AB):=\mathrm{Re}\,U_+(AB)$.
%Analogously to Corollary~\ref{cor-assignment},
To Smirnov's parafermionic percolation observable, one naturally assigns a random variable in the resulting probability space, having the same expectation.
\end{remark}

\section{Lattice Higgs field}

Now we show that the random-medium Ising model itself appears naturally in a linearization of the lattice Higgs field: %Informally,
the former is obtained when the latter assumes only certain basis-vector values.
We follow the popular-science construction of the lattice Higgs field from \cite{Maldacena-16}.

\begin{definition} \label{def-lattice-Higgs}
Let a subset $S\subset G$ be either finite or of positive measure; in the former case equip it with the counting measure. Fix $0\le\beta,\lambda\le1$ and a map $h$ from the set of boundary vertices to the set $S$. %$\{1,i\}$.
For a boundary edge $AB$ set $u(AB)=h(A)h(B)^*$ and assume that $u(AB)\in G$.

Let $S^{V}$ be the set of $S$-valued functions $H$ on vertices which are equal to $h$ on the boundary. %(\emph{Higgs fields}).
Define LatticeHiggs$(G,S,\beta,\lambda,\Omega,h)$ to be the probability space $G^{E}\times S^{V}$ with the probability density
$$
P(U,H)=\frac{1}{Z}\exp\left(\beta\sum_{ABC}%(
\mathrm{Re}\,(U(AB)U(BC)U(CA))%-1)
+\lambda\sum_{AB\in E}\mathrm{Re}\,(H(A)^*U(AB)H(B))
\right).
$$

In the case when $G$ is a group acting on $S$, given a $G$-valued function $g$ on %the set of
vertices, define the \emph{gauge transformation} $G^{E}\times S^{V}\to G^{E}\times S^{V}$ by the formula $U(AB)\mapsto g(A)U(AB)g(B)^*$, $H(A)\mapsto g(A)H(A)$.

Define LinearizedHiggs$(G,S,\beta,\lambda,\Omega,u)$ and LinearizedHiggs$(G,S,0,\infty,\Omega,u)$ to be
the %probability
space $G^{E}\times S^{V}$ with the probability densities respectively
\begin{align*}
P'_{\beta,\lambda,S}(U,H)&=\frac{1}{Z'_{\beta,\lambda,S}}
\prod_{ABC\in F} \left(1+\beta\,\mathrm{Re}\,(U(AB)U(BC)U(CA))\right)
\prod_{AB\in E} \left(1+\frac{1}{2}\lambda\,\mathrm{Re}\,(H(A)^*U(AB)H(B))
\right)^2;\\
P'_{0,\infty,S}(U,H)&=\frac{1}{Z'_{0,\infty,S}}\prod_{AB\in E}
\mathrm{Re}^2\,(H(A)^*U(AB)H(B)).
\end{align*}
%
%For an edge $AB$ oriented from $A$ to $B$ denote $U(BA):=U(AB)^{-1}$.
\end{definition}

\begin{remark}
We square the products over edges because %just linear terms in $\lambda$
they give no contribution to~$Z'$ otherwise.
%Gauge transformation may not be well-defined when $G$ is not a group acting on $S$.
\end{remark}

%\begin{proposition}
%For each triangulated polygon $\Omega$ and each map $\sigma$ from the boundary vertices to $\{+1,-1\}$, the map $(U,\Sigma)\mapsto (U,\sqrt{\Sigma})$, where $\sqrt{-1}:=i$, gives isomorphism of probability spaces
%$$
%\mathrm{RandomMediaIsing}(\Omega,\sigma)\cong
%\mathrm{SimplifiedHiggs}(U(1),\{1,i\},0,\infty,\Omega,\sqrt\sigma).
%$$
%\end{proposition}

\begin{proof}[Proof of the 2nd isomorphism in Proposition~\ref{th-split-intro}] The isomorphism follows directly from Definitions~\ref{def-linearized-gauge} and~\ref{def-lattice-Higgs}:  $P'_{0,\infty,\{1,i\}}(U,H)=P(U,\Sigma)$ identically for $\Sigma=H^2$ because for each $H\in\{1,i\}^V$ and $U\in U(1)^E$ %we have
$$
\mathrm{Re}^2\,(H(A)^*U(AB)H(B))=
\begin{cases}
  \cos^2\theta(AB), & \mbox{if } H(A)=H(B), \\
  \sin^2\theta(AB), & \mbox{if } H(A)\ne H(B).
\end{cases}
\vspace{-0.8cm}
$$
\end{proof}

\begin{example} \label{ex-f'} Consider the hexagon $\Omega$ and the random variable $f$ from Example~\ref{ex-f}.
Let $h(A_n)$ be $1$ for $n = 1,3,4,6$ and be $i$ for $n=2,5$.
Then the crossing probability for site percolation on $\Omega$ %(which  itself equals $1/2$)
equals the expectation of the random variable
$$
f'(\exp(i\theta_1),\dots,\exp(i\theta_6),\exp(i\eta)):=f(\theta_1+\eta,
\theta_2-\eta,\theta_3+\eta,\dots,\theta_6-\eta)
$$
in LinearizedHiggs$(U(1),U(1),0,\infty,\Omega,h)$.
The random variable $f'$ \emph{is} now gauge invariant.
\end{example}

The following proposition shows that for gauge-invariant random variables the choice of $S$ is not essential. This allows to fix a particular value of $H$ instead of summation over all $H$ (just like in Example~\ref{ex-f'} where $f$ was obtained from $f'$ by fixing $\eta$ to zero).

\begin{proposition} \label{l-S}
  Let $G=U(1)$ or $SU(2)$, and $S\subset G$.
  Then for each gauge-invariant integrable function
  $f\colon G^{E}\times G^{V}\to\mathbb{R}$,
  we have
  $$
  %\mathrm{E}f=
  \int_{G^{E}\times S^{V}} f(U,H)P'_{\beta,\lambda,S}(U,H)\,dUdH=
  \int_{G^{E}\times G^{V}} f(U,H)P'_{\beta,\lambda,G}(U,H)\,dUdH.
  $$
\end{proposition}

\begin{proof}
  Fix any $H'\in S^V$. By the gauge invariance, the change of variables $U(AB)=g(A)U'(AB)g(B)^*$, $H(A)= g(A)H'(A)$ takes the right-hand side to
  \begin{multline*}
  \int_{G^{E}\times G^{V}} f(U',H')P'_{\beta,\lambda,G}(U',H')\,dU'dg
  =  \mu(G)^{|V|}\int_{G^{E}} f(U',H')P'_{\beta,\lambda,G}(U',H')\,dU'\\
  =  \mu(G)^{|V|}\frac{Z'_{\beta,\lambda,S}}{Z'_{\beta,\lambda,G}} \int_{G^{E}} f(U',H')P'_{\beta,\lambda,S}(U',H')\,dU'
  =  \mu(S)^{|V|}\int_{G^{E}} f(U',H')P'_{\beta,\lambda,S}(U',H')\,dU'.
  \end{multline*}
  Averaging over $H'\in S^V$ gives the left-hand side.
\end{proof}

\section{Linearization of non-Abelian %and nonassociative
lattice gauge theory}

%Next we extend our results to non-Abelian and nonassociative gauge theory.
Now generalize Example~\ref{ex-small-hexagon} and Definition~\ref{def-linearized-gauge} to non-Abelian and nonassociative gauge theories.

\begin{definition}\label{def-linearized-gauge-general}
%Let $\{\mathrm{e}_1,\dots,\mathrm{e}_n\}$ be the standard basis in $\mathbb{R}^n$. For $n=4,8$ identify it with the sets
Let $G=U(1)$, $SU(2)$ or $S^7$. Let
$S$ be the set $\{1,i\}$, $\{1,i,j,k\}$ or $\{\mathrm{e}_1,\dots,\mathrm{e}_8\}$ of the complex, quaternionic or octonionic units respectively.
%LatticeGaugeTheory$(G,\beta,\Omega,u)$ and SimplifiedGaugeTheory$(G,\beta,\Omega,u)$ are defined literally as in Definition~\ref{def-lattice-gauge}, only $U(1)$ is replaced by $G$ throughout the definition.
Fix $0\le\beta,\lambda\le1$ and a map $h$ from the set of boundary vertices to the set $S$. Recall that
$G^{E}$ is the set of $G$-valued functions $U$ on the set of all edges such that $U(AB)=h(A)h(B)^*$ for each boundary edge $AB$. Define LinearizedGaugeTheory$(G,\beta,\Omega,h)$ and
LinearizedGaugeTheory$(G,\infty,\Omega,h)$
%literally as in Definition~\ref{def-linearized-gauge}, but replace $U(1)$, $\{+1,-1\}$, $P(U,\Sigma)$ by $G$, $S$, $P'_{0,\infty,S}(U,H)$ throughout the definition.
to be the space $G^{E}$ with the probability densities respectively
$$
P''_{\beta,S}(U)=\frac{1}{Z''_{\beta,S}}\left(
1+\beta^fZ'_{0,\infty,S}
\sum_{H\in S^{V}}P'_{0,\infty,S}(U,H)\right)
\quad\text{and}\quad
P''_{\infty,S}(U)=
\sum_{H\in S^{V}}P'_{0,\infty,S}(U,H).
$$
%A random variable $f\colon G^E\to \mathbb{R}$ is \emph{gauge} $S$-\emph{invariant in average}, if $\int_{G^E} f(\phi_g U)P'_{\beta,\lambda,S}(U,H)\, dU$ does not depend on the choice of $H,g\in S$.
%%%A random variable $f\colon G^E\times G^V\to \mathbb{R}$ is \emph{gauge} $S$-\emph{invariant in average}, if $\int_{G^E} f(U,H')P'_{\beta,\lambda,S}(U,H)\, dU$ does not depend on the choice of $H,H'\in S$.
%A \emph{coordinate reflection} $U(1)^E\to U(1)^E$ or $SU(2)^E\to SU(2)^E$ is a reflection in a coordinate hyperplane of space $\mathbb{R}^{2{|E|}}=\mathbb{C}^{|E|}$ or $\mathbb{R}^{4{|E|}}=\mathbb{H}^{|E|}$ containing $G^E$ respectively. (Their number is $2{|E|}$ or $4{|E|}$.)
A \emph{coordinate reflection} $G^E\to G^E$ is a reflection in a coordinate hyperplane of space $\mathbb{R}^{2{|E|}}=\mathbb{C}^{|E|}$, $\mathbb{R}^{4{|E|}}=\mathbb{H}^{|E|}$, or $\mathbb{R}^{8{|E|}}=\mathbb{O}^{|E|}$ containing $G^E$. %(E.g., there are $2$ coordinate reflections $U(1)\to U(1)$.)
%A \emph{coordinate reflection} $G^E\to G^E$ is a map acting as a coordinate reflection on a single $G$-factor, and acting identically on the others.
\end{definition}

%\begin{example} The random variable $f$ from Example~\ref{ex-f} is gauge $\{1,i\}$-invariant in average for $\beta=\infty$, $\lambda=0$ by Example~\ref{ex-small-hexagon}.
%\end{example}

%\begin{proposition} \mscomm{???}
%  Let $G=U(1)$ or $SU(2)$, and let $S$ be $\{1,i\}$ or $\{1,i,j,k\}$ respectively. Let an integrable function $f\colon G^{E}\to\mathbb{R}$, be gauge $S$-invariant in average and invariant under conjugation or changing the sign of any single coordinate.
%  %  such that
%  %$$
%  %f(x_1,\dots,x_n,\dots,x_{e+v-2b})=
%  %f(x_1,\dots,-x_n,\dots,x_{e+v-2b})=
%  %f(x_1,\dots,x_n^*,\dots,x_{e+v-2b})
%  %$$
%  %for each $n\le e$ and all $x_1,\dots,x_n,\dots,x_{e+v-2b}\in U(1)$,
%  Then
%  $$
%  %\mathrm{E}f=
%  \int_{G^{E}\times G^{V}} f(\phi_{H^*}U)P'_{\beta,\lambda,G}(U,H)\,dUdH=
%  \int_{G^{E}}f(U)P''_{\beta,S}(U)P''_{\lambda^2/4,S}(U)\,dU.
%  $$
%\end{proposition}
%
%\begin{proof}
%Notice that the function $f(\phi_{H^*}U)$ in $U$ and $H$ is gauge invariant. By Lemma~\ref{l-S}, gauge $S$-invariance in average, and Definition~\ref{def-lattice-Higgs}, the left-hand side equals
%$$
%%\int_{G^{E}\times S^{V}} f(U)P'_{\beta,\lambda}(U,H)\,dUdH.
%\frac{1}{Z'_{\beta,\lambda,S}}
%\int_{G^{E}\times S^{V}} f(U)
%\prod_{ABC}\left(1+\beta\,\mathrm{Re}\,(U(AB)U(BC)U(CA))\right)
%\prod_{AB}\left(1+\frac{1}{2}\lambda\,\mathrm{Re}\,(H(A)^*U(AB)H(B))
%\right)^2\,dUdH
%$$
%\end{proof}

\tmpcomment

%We need the following simple and well-known lemma; the proof is omitted.

%%% !!! Wrong
%\begin{lemma}[Graph coloring interpretation of 4-state Potts model] The following map is a bijection between all $4$-colorings of nonboundary vertices of $\Omega$ and all regularly edge-$3$-colored subgraphs of $\Omega$ without isolated vertices: Label the $4$ vertex colors (respectively, the $3$ edge colors) by the elements (respectively, nonzero elements) of $\mathbb{Z}/2\mathbb{Z}\oplus \mathbb{Z}/2\mathbb{Z}$. Fix an arbitrary coloring of the boundary vertices. Take a $4$-coloring $\phi\in(\mathbb{Z}/2\mathbb{Z}\oplus \mathbb{Z}/2\mathbb{Z})^V$. Include an edge $AB$ into the desired subgraph, if $\phi(A)+\phi(B)\ne 0$, and then paint it in the color $\phi(A)+\phi(B)$.
%\end{lemma}

\endtmpcomment

Let us restate Proposition~\ref{prop-linearization-intro} in a slightly more general form.

\begin{proposition}\label{prop-linearization}
  Let $G=U(1)$ or $SU(2)$, %, or $S^7$,
  $S=\{1,i\}$ or $\{1,i,j,k\}$ %or $\{\mathrm{e}_1,\dots,\mathrm{e}_8\}$
  respectively. If an integrable function  $f\colon G^{E}\to\mathbb{R}$ is invariant under each coordinate reflection $G^{E}\to G^{E}$, then
  $$
  %\mathrm{E}f=
  \int_{G^{E}} f(U)P'(U)\,dU=
  \int_{G^{E}} f(U)P_{\beta,S}''(U)\,dU.
  $$
%%%%
%For each integrable $f\in L_1(LinearizedGaugeTheory$(U(1),\beta,\Omega,\sigma))$ invariant under all transformations $\theta()\!\mapsto\!\pi\!\pm\!\theta_n\!\!\!\!\!\!\pmod{2\pi}$ we have}
\end{proposition}

%\begin{problem} Give a counterexample to the analogue of the proposition for $G=S^7$. % and $S=\{\mathrm{e}_1,\dots,\mathrm{e}_8\}$.
%\end{problem}

We present a proof for $G=SU(2)$; the one for $U(1)$ is simpler and is easily obtained from that. The argument is similar to Example~\ref{ex-small-hexagon} but uses regularly edge-$3$-colored graphs instead of broken lines.

\begin{lemma} \textup{(Cf.~\cite[p.~235]{Penrose-71})} \label{l-plus}
Let $h$ be a map from the set of boundary vertices of $\Omega$ to $SU(2)$. Take $U\in SU(2)^E$ such that $U(AB)=h(A)h(B)^*$ for each boundary edge $AB$ and $U(AB)U(BC)U(CA)=\pm 1$ for each face $ABC$. Then $\Pi_{ABC}U(AB)U(BC)U(CA)=1$.
\end{lemma}

\begin{proof}[Proof of Lemma~\ref{l-plus}]
  Use induction over the number of faces. If $\Omega$ has a single face $ABC$, then $$U(AB)U(BC)U(CA)=h(A)h(B)^*h(B)h(C)^*h(C)h(A)^*=1.$$
  Otherwise let $ABC$ be a face such that the edge $AB$ is on the boundary and $BC$ is not on the boundary. If $U(AB)U(BC)U(CA)=-1$ then change the sign of $U(BC)$; this does not affect the product over all faces in question. We get $U(AB)U(BC)U(CA)=+1$. Set $h(C)=U(CA)h(A)$. Then $U(CA)=h(C)h(A)^*$ and $$U(BC)=U(AB)^*\cdot U(AB)U(BC)U(CA)\cdot U(CA)^*=h(B)h(A)^*\cdot 1\cdot h(A)h(C)^*=h(B)h(C)^*.$$ Remove the triangle $ABC$ from $\Omega$. Applying the inductive hypothesis to the remaining polygon(s), we arrive at the required assertion.
\end{proof}

\begin{remark}\label{ex-octonion}
  Lemma~\ref{l-plus} does not hold for $SU(2)$ replaced by $S^7$. E.g., if $\Omega$ has a single face $ABC$, $h(A)=ab,h(B)=b,h(C)=c$, where $a,b,c$ are octonion units satisfying $(ab)c=-a(bc)$, then
  $(U(AB)U(BC))U(CA)=(((ab)b^*)(bc^*))(c(ab)^*)=(-c(ab))^2=-1$.
\end{remark}

\begin{proof}[Proof of~Propositions~\ref{prop-linearization-intro} and~\ref{prop-linearization} for $G=SU(2)$]
The propositions follow from the set of equalities
\begin{align*}
Z'\int_{G^{E}} f(U)P'(U)\,dU&=
\int_{G^{E}} f(U)\,dU\,\prod_{ABC}\left(1+\beta\,
\mathrm{Re}\,(U(AB)U(BC)U(CA))\right)\\
&=
\int_{G^{E}} f(U)\,dU\,\prod_{ABC}\left(1+\beta\,
\sum_{a,b,c\in S: abc=\pm 1} abc U_{a}(AB)U_b(BC)U_c(CA)\right)\\
&=
\int_{G^{E}} f(U)\,dU\,\left(1+\beta^{|F|}\,
\sum_{\substack{
u\in S^{E}:\\
u(AB)u(BC)u(CA)=\pm 1
%\\u=U\text{ on }\partial\Omega
}}
\prod_{ABC}u(AB)u(BC)u(CA)\prod_{AB\in E}U_{u(AB)}(AB)^2\right)\\
&=
\int_{G^{E}} f(U)\,dU\,\left(1+\beta^{|F|}\,
\sum_{\substack{
u\in S^{E}:\\
u(AB)u(BC)u(CA)=\pm 1
%\\u=U\text{ on }\partial\Omega
}}
\prod_{AB\in E}U_{u(AB)}(AB)^2\right)\\
%\intertext{\ }
&=
\int_{G^{E}} f(U)\,dU\,\left(1+\beta^{|F|}\,\sum_{H\in S^V}\prod_{AB\in E}\mathrm{Re}^2\,(H(A)^*U(AB)H(B))\right)\\
&=
Z_{\beta,S}''\int_{G^{E}} f(U)P_{\beta,S}''(U)\,dU.
\end{align*}
Setting $f(U)=1$, we get $Z'=Z_{\beta,S}''$, and we are done. It remains to explain the equalities.

Here the first and the last equalities follow from Definitions~\ref{def-lattice-gauge} and~\ref{def-linearized-gauge-general}.

In the 2nd equality we use basis decomposition
 $U(AB)=:U_1(AB)+iU_i(AB)+jU_j(AB)+kU_k(AB)$.

In the 3rd equality the summation is over all functions $u$ on edges such that
\begin{itemize}
\item $u(AB)=U(AB)=h(A)h(B)^*\in Q_8:=\{\pm 1,\pm i,\pm j,\pm k\}$ for each boundary edge $AB$,
\item $u(AB)\in S=\{1,i,j,k\}$ for each nonboundary edge $AB$, and
\item $u(AB)u(BC)u(CA)=\pm 1$ for each face $ABC$.
\end{itemize}
The sum is obtained by expanding the product in the previous line. By the coordinate reflection invariance, the terms containing the first power of $U_{a}(AB)$ for some nonboundary edge $AB$ do not contribute to the integral. Hence only the product of square factors $U_{a}(AB)^2$ over all nonboundary edges $AB$ survives (and the free term). For a boundary edge $AB$, the factor $aU_{a}(AB)$
is nonzero, only if $a=\pm U(AB)$ and hence $aU_{a}(AB)=u(AB)$.

%The second product is over nonboundary edges $AB$.

The 4th equality follows from Lemma~\ref{l-plus}.
% because a nonzero contribution comes only from $u$ such that $u(AB)=h(A)h(B)^*$ for all boundary edges $AB$ (otherwise $U_{u(AB)}(AB)=0$).

In the 5th equality, to every $u\in S^E$ such that $u(AB)u(BC)u(CA)=\pm 1$ for each face $ABC$,
we assign the unique $H\in S^V$ such that $u(AB)=\pm H(A)H(B)^*$ for each edge $AB$ (and $H(A)=h(A)$ for each boundary vertex $A$). This is possible by a standard ``homological'' argument (the existence of~$H$ is equivalent to the assertion $H^1(\Omega;\mathbb{Z}/2\mathbb{Z}\oplus \mathbb{Z}/2\mathbb{Z})=0$ because $Q_8/\{\pm 1\}\cong \mathbb{Z}/2\mathbb{Z}\oplus \mathbb{Z}/2\mathbb{Z}$). Then $U_{u(AB)}(AB)=\pm \mathrm{Re}\,(H(A)^*U(AB)H(B))$, which completes the proof.
\end{proof}

%\begin{definition} \mscomm{???}
%For any $\beta,\lambda>0$ define LinearizedHiggs$(U(1),\beta,\lambda,\Omega,h)$ to be the probability space $U(1)^{E}\times U(1)^{V}$ with the probability density %ies respectively
%$$
%P''_{\beta,\lambda}(U,H)=\frac{1}{Z'''_{\beta,\lambda}}
%\left(1+\beta^{|F|}
%\sum_{\Sigma\in\{1,i\}^{V}}P'_{\beta,\lambda}(\phi_H U,\Sigma)\right)
%\left(1+\left(\frac{\lambda}{2}\right)^{2({|E|})}
%\sum_{\Sigma\in\{1,i\}^{V}}P'''(\phi_H U,\Sigma)\right).
%%\\
%%P''_{0,\infty}(U,H)&=\frac{1}{Z'''_{0,\infty}}\sum_{\Sigma\in\{+1,-1\}^{V}}P(\phi_H U,\Sigma).
%%\end{align*}
%$$
%\end{definition}

%\begin{proposition} \label{th-Higgs-split} \mscomm{???}
%For each triangulated polygon $\Omega$ and each map $h$ from the boundary vertices to $\{1,i\}$,
%$$
%\mathrm{LinearizedHiggs}(U(1),\infty,\infty,\Omega,h))=
%\mathrm{RandomMediaIsing}(\Omega,h^2)\times U(1)^{V}
%$$
%as probability spaces.
%\end{proposition}
%

\begin{definition} %Let $S$ have $s$ elements.
%Let $\{\mathrm{e}_1,\dots,\mathrm{e}_n\}$ be the standard basis in $\mathbb{R}^n$.
%Fix a map $h$ from the set of boundary vertices to $\{\mathrm{e}_1,\dots,\mathrm{e}_n\}$.
Define Potts$(|S|,\infty,\Omega,h)$ to be the space $S^V$ with counting measure divided by $|S|^{|V|}$.
\end{definition}

\begin{proof}[Proof of the 3rd isomorphism in Proposition~\ref{th-split-intro}] This is straightforward:
%follows directly from the definitions.
Let $\nu$ and $\mu$ be the probability measures in $\mathrm{LinearizedGaugeTheory}(SU(2),\infty,\Omega,h)$ and
$\mathrm{LinearizedHiggs}(SU(2),S,\infty,\Omega,h)$ respectively, and $\lambda$ be the Lebesque measure on $SU(2)^E\times S^{V}$, where $S=\{1,i,j,k\}$.
Then for each measurable $A\subset SU(2)^E$
$$
\nu(A)=\int_A P''_{\infty,S}(U)\,dU
=\int_A\,
\sum_{H\in S^{V}}P'_{0,\infty,S}(U,H)\,dU
=\int_{A\times S^{V}}P'_{0,\infty,S}(U,H)\,d\lambda
=\mu(A\times S^{V}).
$$
%By Definition~\ref{def-quotient} we get the 3rd isomorphism in Proposition~\ref{th-split-intro}.
%\end{proof}
%
%\begin{remark}
Since $\int_{SU(2)^{E}}P'_{0,\infty,\{1,i,j,k\}}(U,H)\,dU$
is the same
%=\pi^{2{|E|}}/2^{{|E|}}$
for each $H$, %\in\{1,i,j,k\}^{V}$,
we also have $\mu(SU(2)^E\times B)={|B|}/{4^{|V|}}$.
%$$
%\mathrm{Potts}(4,\infty,\Omega,h)\cong
%\mathrm{LinearizedHiggs}(SU(2),\{1,i,j,k\},\infty,\Omega,h)
%/\mathrm{LinearizedGaugeTheory}(SU(2),\infty,\Omega,h).
%$$
%\end{remark}
\end{proof}

%\tmpcomment

\section{Variations}

In a sense, Proposition~\ref{prop-linearization} holds for
$G=S^2=\{z\in\mathbb{H}:|z|=1,\mathrm{Re}\,z=0\}$. Surprisingly, although $S^2$ lives in $3$-dimensional space, we still have the Potts model with $4$ colors, but now antiferromagnetic.

\begin{definition} %Let $S$ have $s$ elements.
%Let $\{\mathrm{e}_1,\dots,\mathrm{e}_n\}$ be the standard basis in $\mathbb{R}^n$.
%Fix a map $h$ from the set of boundary vertices to $\{\mathrm{e}_1,\dots,\mathrm{e}_n\}$.
Define AntiferromagneticPotts$(|S|,0,\Omega,h)$ to be the subspace of Potts$(|S|,\infty,\Omega,h)$ formed by $H\in S^V$ such that $H(A)\ne H(B)$ for each edge $AB$.
%Define {AntiferromagneticHiggs}$(G,S,0,\infty,\Omega,h)$ analogously.
\end{definition}

\begin{proposition} \label{th-split-S2}
For each polygon $\Omega$ triangulated by regular triangles, each map $h$ from the set of boundary vertices to $\{1,i,j,k\}$ such that $h(A)\ne h(B)$ for each edge $AB$, and each function $f\colon (S^2)^E\to\mathbb{R}$ invariant under each coordinate reflection of $\mathbb{R}^{3|E|}\supset (S^2)^E$, we have
\begin{multline*}
  \int\limits_{(S^2)^{E}} dU\,f(U)\prod_{ABC\in F} \left(1+\beta
  \,\mathrm{Re}\,(U(AB)U(BC)U(CA))
  \right)%\,\prod_{AB}dU(AB)
  =\\=
  \int\limits_{(S^2)^{E}}  dU\,f(U)\left(1+\beta^{|F|}\,
  \sum_{H\in \mathrm{AntiferromagneticPotts}(4,0,\Omega,h) }\prod_{AB\in E} \mathrm{Re}^2\,(H(A)^*U(AB)H(B))\right),
  %%\,\prod_{AB}dU(AB).
  %%%%%%%%%%
  %\hspace{-0.8cm}\int_{G^{E}} dUf(U)\prod_{ABC}\left(1+\beta
%\,\mathrm{Re}\,(U(AB)U(BC)U(CA))
%\right)%\,\prod_{AB}dU(AB)
  %&=
  %\int_{G^{E}}  dUf(U)\left(1+\beta^{|F|}\,\sum_{H\in S^V}\prod_{AB}\mathrm{Re}^2\,(H(A)^*U(AB)H(B))\right)
  %%\,\prod_{AB}dU(AB).
  %\end{align*}
  \end{multline*}
%%%
%\begin{multline*}
%\mathrm{LinearizedGaugeTheory}(S^2,\infty,\Omega,h))\ltimes\mathrm{AntiferromagneticPotts}(4,\Omega,0,h)\\
%\cong
%\mathrm{AntiferromagneticHiggs}(S^2,\{1,i,j,k\},0,\infty,\Omega,h).
%\end{multline*}
%%%
\end{proposition}

%\begin{remark} Although $S^2$ lives in $3$-dimensional space, we still have the Potts model with $4$ colors.
%\end{remark}

\begin{proof} The proof is the same as that of Proposition~\ref{prop-linearization}, %with the following modifications.
but $S$ should be replaced by $\{i,j,k\}$ in the 2nd and the 3rd equality, whereas  $S^V$ should be replaced by {AntiferromagneticPotts}$(4,0,\Omega,h)$ in the 5th equality,
because $\pm H(A)H(B)^*=u(AB)\in\{i,j,k\}$ implies that $H(A)\ne H(B)$.
%
%The proof of the isomorphism is analogous to the last isomorphism in Proposition~\ref{th-split-intro}.
\end{proof}

%\endtmpcomment

%\begin{theorem}
%There is the following commutative diagram, where the map $\mathrm{E}(\cdot | \text{Colors})$ is surjective:
%$$
%\begin{CD} L_1(\mathrm{RandomMediaPotts})
%@>\mathrm{E}(\cdot | \text{Weights})>>
%L_1(\text{Lattice $SU(2)$-gauge theory})\\
%@V\mathrm{E}(\cdot | \text{Colors}) VV
%@V\left.\frac{\partial^n}{\partial\beta^n}\mathrm{E}(\cdot)\right|_{\beta=0}VV\\
%L_1(\text{4-State Potts})
%@>\mathrm{E}>>
%\mathbb{R}.
%\end{CD}
%$$
%\end{theorem}

%\begin{theorem}
%There is the following commutative diagram, where the map $\mathrm{E}(\cdot | \text{Colors})$ is surjective:
%$$
%\begin{CD} L_1(\text{Random-Media 8-State Potts})
%@>\mathrm{E}(\cdot | \text{Weights})>>
%L_1(\text{Lattice octonion-gauge theory})\\
%@V\mathrm{E}(\cdot | \text{Colors}) VV
%@V\left.\frac{\partial^n}{\partial\beta^n}\mathrm{E}(\cdot)\right|_{\beta=0}VV\\
%L_1(\text{8-State Potts})
%@>\mathrm{E}>>
%\mathbb{R}.
%\end{CD}
%$$
%\end{theorem}

%\section{Open problems}
%\label{sec-open}

\begin{figure}[htbp]
\centering
\includegraphics[width=2.7cm]{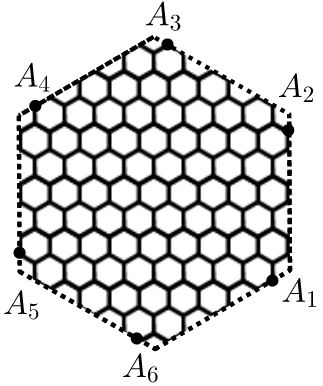}%\\
\caption{Simultaneous crossing; see Conjecture~\ref{conj-independence}}
\label{fig1}
\end{figure}
The above results suggest that the $4$-state Potts model might have a special interest. In that model, the simplest event not ``reducing'' to $2$-state ones is the following ``simultaneous crossing''.

\begin{conjecture}\label{conj-independence} (see Figure~\ref{fig1})
Take a polygon $\Omega$ with the boundary composed of the edges of a hexagonal lattice of mesh $h$. Let $A_1,A_2,A_3,A_4,A_5,A_6$ be some boundary vertices. To each hexagon of the lattice, assign one of the $4$ colors independently with probability $1/4$. Consider the $3$ events
\begin{align*}
E_1&=(\text{$A_1A_2$ and $A_4A_5$ are joined by a path formed by hexagons of colors 1 or 4}),\\
E_2&=(\text{$A_2A_3$ and $A_5A_6$ are joined by a path formed by hexagons of colors 2 or 4}),\\
E_3&=(\text{$A_3A_4$ and $A_6A_1$ are joined by a path formed by hexagons of colors 3 or 4}).
\end{align*}
Then these $3$ events become mutually independent in the limit when $h\to 0$ and $(\Omega, A_1,A_2,A_3,A_4,A_5,A_6)$ approaches a planar domain (having rectifiable Jordan boundary) with $6$ distinct boundary points, in the Caratheodory sense. That is, $P(E_1 \cap E_2 \cap E_3) - P(E_1)P(E_2)P(E_3) \to 0$ under this limit.
%(B) The same but in $E_2$ we take the arcs $A_2A_3$ and $A_6A_1$, while in $E_3$ we take $A_3A_4$ and $A_5A_6$.
\end{conjecture}

\begin{remark}
Informally, there are 3 percolating fluids, each hexagon is open either for just one fluid or for three altogether, and we study simultaneous percolation of the fluids.

Notice that the 3 events are obviously pairwise independent but not mutually independent in general (take $\Omega$ to be just one hexagon). There is some numerical evidence for the conjecture.
%So far able to prove only that A) and B) are equivalent.
\end{remark}

% TOTALLY INCORRECT
%\begin{conjecture}
%  Consider the spin Ising model on a regular triangular lattice with random edge weights $\cot^2 x_1,\dots,\cot^2 x_n$, where $x_1,\dots,x_n$ are taken independently uniformly from $(0,\pi)$. Then for each topological quadrilateral the crossing probability in this model equals to the crossing probability in site percolation on the same lattice.
%\end{conjecture}
%
%\begin{remark}
%  This is true without passing to the continuum limit.
%  Moreover, the model splits as a direct product of site percolation and ``linearized'' $U(1)$ lattice gauge theory on the same lattice.
%\end{remark}

%\begin{thebibliography}{9}

%\bibitem{Duminil-Copin-etal-14} H.~Duminil-Copin, R.~Peled, W.~Samotij, Y.~Spinka, Exponential decay of loop lengths in the loop O(n) model with large n, arXiv:1412.8326v3.

%\bibitem{Ootsuka-etal-05} T. Ootsuka, E. Tanaka, E. Loginov, Non-associative Gauge Theory, 2005, \href{https://arxiv.org/abs/hep-th/0512349}{arXiv:hep-th/0512349v2}

%\end{thebibliography}

%\end{document}

%\tmpcomment

\subsection*{Acknowledgements}

For the latter conjecture, there have been suggested a proof by K.~Izyurov and A.~Magazinov, as well as interesting generalizations by %E.~Akhmedova,
M.~Fedorov and I.~Novikov (private communication)  \cite{Fedorov-19, Novikov-19}.
%\textbf{Acknowledgements.}
The author is grateful to
%M.~Fedorov,
D.~Chelkak,
H. Duminil-Copin,
%K.~Izyurov,
M.~Khristoforov,
%A.~Magazinov,
and
%, I.~Novikov,
S.~Smirnov for useful discussions.
%  E.Akhmedov, L.Alania, D.Arnold, A.Bossavit, V.Buchstaber, M.Chernodub, M.Desbrun, F.G\"unther, I.Ivanov, M.Kraus, N.Mnev, F.M\"uller-Hoissen, S.Pirogov, R.Rogalyov, I.Sabitov, P.Schr\"oder, I.Shenderovich, B.Springborn, A.Stern, S.Tikhomirov, S.Vergeles, for useful discussions.

%Funding information pending.

\small
\vspace{-0.2cm}
\noindent
\textsc{Mikhail Skopenkov\\
National Research University Higher School of Economics (Faculty of Mathematics) \&\\
%and\\
Institute for Information Transmission Problems, Russian Academy of Sciences} %\\
%Bolshoy Karetny per. 19, bld. 1, Moscow, 127994, Russian Federation
\\
\texttt{mikhail.skopenkov\,@\,gmail$\cdot $com} \quad \url{http://www.mccme.ru/~mskopenkov}

\end{document}